\documentclass[letterpaper,USenglish]{oasics-v2018} 

\usepackage{xspace}%
\usepackage{microtype}
\usepackage{xcolor}%
\usepackage{float}%
\usepackage{tikz}%

\usepackage[symbol]{footmisc}%
\title{Subset Sum Made Simple} 

\titlerunning{Subset Sum Made Simple}
 
\author{Konstantinos Koiliaris}{Department of Computer Science, University of Illinois Urbana-Champaign\\{201 North Goodwin Avenue, Urbana, IL 61801, USA}}{koiliar2@illinois.edu}{https://orcid.org/0000-0002-1842-1829}{}

\author{Chao Xu}{
\emph{Current affiliation:} Yahoo! Research \\{770 Broadway, 6th Floor, New York, NY 10003, USA}\\
\emph{Previous affiliation:} Department of Computer Science, University of Illinois Urbana-Champaign\\{201 North Goodwin Avenue, Urbana, IL 61801, USA}}{chao.xu@oath.com}{https://orcid.org/0000-0003-4417-3299}{}

\authorrunning{K. Koiliaris and C. Xu}

\Copyright{Konstantinos Koiliaris and Chao Xu}

\subjclass{Theory of computation $\rightarrow$ Algorithm design techniques}

\hideOASIcs
\category{}

\relatedversion{}

\supplement{}

\funding{}

\acknowledgements{We would like to thank Sariel Har-Peled for his invaluable help in the editing of this paper.}

\EventEditors{}
\EventNoEds{2}
\EventLongTitle{42nd Conference on Very Important Topics (CVIT 2016)}
\EventShortTitle{CVIT 2016}
\EventAcronym{CVIT}
\EventYear{2016}
\EventDate{December 24--27, 2016}
\EventLocation{Little Whinging, United Kingdom}
\EventLogo{}
\SeriesVolume{42}
\ArticleNo{1}  

\nolinenumbers 
\hideOASIcs  

\newcommand{\set}[1]{\left\{\, #1 \,\right\}}%
\newcommand{\FFT}{\textsf{F{F}T}\xspace}%
\newcommand{\SSSu}[1]{\mathcal{S}_{u}\bigl(#1\bigr)}%
\newcommand{\SSSCu}[1]{\mathcal{S}^{\#}_{u}\bigl(#1\bigr)}%
\newcommand{\SSSCub}[1]{\mathcal{S}^{\#}_{u/b}\bigl(#1\bigr)}%

\begin{document} 
 
\maketitle

\begin{abstract}
    \textsc{SubsetSum} is a classical optimization problem taught to undergraduates as an example of an NP-hard problem, which is amenable to dynamic programming, yielding polynomial running time if the input numbers are relatively small. Formally, given a set $S$ of $n$ positive integers and a target integer $t$, the \textsc{SubsetSum} problem is to decide if there is a subset of $S$ that sums up to $t$. Dynamic programming yields an algorithm with running time $O(nt)$. Recently, the authors \cite{koiliaris2017faster} improved the running time to $\tilde{O}\bigl(\sqrt{n}t\bigr)$, and it was further improved to $\tilde{O}\bigl(n+t\bigr)$ by a somewhat involved randomized algorithm by Bringmann \cite{bringmann2017near}, where $\tilde{O}$ hides polylogarithmic factors. 

    Here, we present a new and significantly simpler algorithm with running time $\tilde{O}\bigl(\sqrt{n}t\bigr)$. While not the fastest, we believe the new algorithm and analysis are simple enough to be presented in an algorithms class, as a striking example of a divide-and-conquer algorithm that uses \FFT to a problem that seems (at first) unrelated. In particular, the algorithm and its analysis can be described in full detail in two pages (see pages~\pageref{alg:start}--\pageref{alg:end}).
\end{abstract}

\section{Introduction}
    Given a (multi) set $S$ of $n$ positive integers and an integer target value $t$, the \textsc{SubsetSum} problem is to decide if there is a (multi) subset of $S$ that sums up to $t$. The \textsc{SubsetSum} is a classical problem with relatively long history.
    It is one of Karp's original NP-complete problems \cite{Karp1972}, closely related to other fundamental NP-complete problems such as \textsc{Knapsack} \cite{Dantzig1957}, \textsc{Constrained Shortest Path} \cite{AxiotisT18}, and various other graph problems with cardinality constraints \cite{Eppstein1997375,guruswami,NET:NET5}. 
    Furthermore, it is one of the initial \emph{weakly} NP-complete problems; problems that admit \emph{pseudopolynomial} time algorithms -- a classification identified by Garey and Johnson in \cite{garey1978strong}. The first such algorithm was given in 1957\,\footnote[1]{~Note that Bellman wrote this paper before the definition of pseudopolynomial time algorithms was provided by Garey and Johnson in 1977.} by Bellman, who showed how to solve the problem in $O(nt)$ time using dynamic programming \cite{Bellman1956}.   

    The importance of the \textsc{SubsetSum} problem in computer science is further highlighted by its role in teaching. Both the problem and its algorithm have been included in undergraduate algorithms courses' curriculums and textbooks for several decades (\cite[Chapter~34.5.5]{Cormen:2009:IAT:1614191}, used as archetypal examples for introducing the notions of weak NP-completeness and pseudopolynomial time algorithms to college students \cite[Chapter~8.8]{Kleinberg:2005:AD:1051910}. In addition, the conceptually simple problem statement makes this problem a great candidate in the study of NP-completeness \cite[Chapter~8.1]{Dasgupta:2006:ALG:1177299}), and, finally, Bellman's algorithm is also often introduced in the context of teaching dynamic programming \cite[Chapter~5.6]{jeffe2015notes}. 
    
    Extensive work has been done on finding better and faster pseudopolynomial time algorithms for the \textsc{SubsetSum} (for a collection of previous results see \cite[Table~1.1]{koiliaris2017faster}). The first improvement on the running time was a $O(nt/\log t)$ time algorithm by \cite{Pisinger19991}, almost two decades go. 
    Recently, the \emph{state-of-the-art} was improved significantly to $\tilde{O}\bigl(\sqrt{n}t\bigr)$ time by the authors \cite{koiliaris2017faster}.
    Shortly after, in a follow up work, the running time was further improved to $\tilde{O}(n+t)$ time by Bringmann \cite{bringmann2017near} -- the algorithm is randomized and somewhat involved. 
    Abboud et al. \cite{AbboudBHS17} showed that it is unlikely that any \textsc{SubsetSum} algorithm runs in time $O\bigl(t^{1-\varepsilon}\,2^{o(n)}\bigr)$, for any constant $\varepsilon>0$ and target number $t$, as such an algorithm would imply that the Strong Exponential Time Hypothesis (SETH) of Impagliazzo and Paturi \cite{ImpagliazzoP01} is false.

    In this paper, we present a new \emph{simple} algorithm for the \textsc{SubsetSum} problem. The algorithm follows the divide-and-conquer paradigm and uses the Fast Fourier Transform (\FFT), matching the best deterministic running time $\tilde{O}\bigl(\sqrt{n}t\bigr)$ of \cite{koiliaris2017faster} with a cleaner and more straightforward analysis. The algorithm partitions the input by congruence into classes, computes the subset sums of each class recursively, and combines the results. We believe this new simple algorithm, although not improving upon the state-of-the-art, reduces the conceptual complexity of the problem and improves our understanding of it. We believe the new algorithm can be used in teaching as an example of a pseudopolynomial time algorithm for the \textsc{SubsetSum} problem, as well as a striking example of applying \FFT to a seemingly unrelated problem.

    \paragraph*{Comparison to previous work}
    Our previous algorithm \cite{koiliaris2017faster} used a more complicated divide-and-conquer strategy that resulted in forming sets of two different types, that had to be handled separately. Bringmann's algorithm \cite{bringmann2017near} uses randomization and a two-stage color-coding process. Both algorithms are significantly more complicated than the one presented here.

\section{Preliminaries}
    Let $[u] = \bigl\{\,0, 1, \ldots, \lceil u\rceil\,\bigr\}$ denote the set of integers in the interval $\bigl[0, \lceil u\rceil\bigr]$. Given a set $X \subset \mathbb{N}$, let $\Sigma X = \sum_{x \in X} x$ and denote the \emph{set of all subset sums of $X$ up to $u$} by 
    \[  
        \SSSu{X} = \bigl\{\, \Sigma Y \bigm| Y \subseteq X \,\bigr\} \cap [u], 
    \]
    and the \emph{set of all subset sums of $X$ up to $u$ with cardinality information} by
    \[
        \SSSCu{X} = \bigl\{\, \bigl(\Sigma Y, |Y|\bigr) \bigm| Y \subseteq X \,\bigr\} \cap \bigl([u] \times \mathbb{N}\bigr).
    \]
    Let $X$, $Y$ be two sets, the \emph{set of pairwise sums of $X$ and $Y$ up to $u$} is denoted by
    \[
        X \oplus_u Y = \set{x+y \mid x \in X,\, y \in Y} \cap [u].
    \]
    If $X$, $Y \subseteq \mathbb{N} \times \mathbb{N}$ are sets of points in the plane, then
    \[
        X \oplus_u Y = \set{(x_1+y_1, x_2+y_2) \mid (x_1,x_2)\in X,\, (y_1,y_2)\in Y } \cap \bigl([u] \times \mathbb{N}\bigr).
    \]

    Observe, that if $X$ and $Y$ are two disjoint sets, then $\SSSu{X \cup Y} = \SSSu{X} \oplus_u \SSSu{Y}$. 

    Next, we define two generalizations of the \textsc{SubsetSum} problem. Both can be solved by the new algorithm.
          
    \begin{figure}[ht]
        \begin{tikzpicture}
            \draw node[fill=black!12,inner sep=1ex,text width=(\textwidth/2)-2.5ex] {
                \textsc{AllSubsetSums}\\
                \vspace*{.21cm}
                \small
                {\small\textsf{INPUT:}} Given a set $S$ of $n$ positive integers and an upper bound integer $u$. \\
                \vspace*{.20cm}
                {\small\textsf{OUTPUT:}} The set of \emph{all realizable subset sums} of $S$ up to $u$. 
                \vspace*{.25cm}};
        \end{tikzpicture}
        \begin{tikzpicture}
            \draw node[fill=black!12,inner sep=1ex,text width=(\textwidth/2)-2.5ex] {
                \textsc{AllSubsetSums}$^{\#}$\\
                \vspace*{.15cm}  
                \small
                {\small\textsf{INPUT:}} Given a set $S$ of $n$ positive integers and an upper bound integer $u$. \\
                {\small\textsf{OUTPUT:}} The set of all realizable subset sums along with \emph{the size of the subset that realizes each sum} of $S$ up to $u$.}; 
        \end{tikzpicture}
        \caption{Two generalizations of the \textsc{SubsetSum} problem.}
        \label{algorithm} 
    \end{figure}

    Note that the case where the input is a multiset can be reduced to the case of a set with little loss in generality and running time (see \cite[Section 2.2]{koiliaris2017faster}), hence for simplicity of exposition we assume the input is a \emph{set} throughout the paper. 

\section{The algorithm} 
\label{alg:start} 
    Here, we show how to solve \textsc{AllSubsetSums} in $\tilde{O}\bigl(\sqrt{n}t\bigr)$ time. Clearly, computing all subset sums up to $u$ also decides \textsc{SubsetSum} with target value $t \leq u$.
    
    \subsection{Building blocks}

        The following well-known lemma describes how to compute pairwise sums between sets in almost linear time, in the size of their ranges, using \FFT.

        \begin{lemma}[Computing pairwise sums $\oplus_u$]
            \label{lem:l1}
            The following are true:
            \begin{enumerate}[(A)]
                \item Given two sets $S$, $T \subseteq [u]$, one can compute $S \oplus_u T$ in $O( u \log u )$ time. \label{(A)}
                \item Given $k$ sets $S_1, \dots, S_k \subseteq [u]$, one can compute $S_1 \oplus_u \dots \oplus_u S_k$ in $O( k\,u \log u )$ time. \label{(B)}
                \item Given two sets of points $S$, $T \subseteq [u] \times [v]$, one can compute $S \oplus_u T$ in $O\bigl( u\,v \log (u\,v)\bigr)$ time. \label{(C)} 
            \end{enumerate}
        \end{lemma}
        \begin{proof}
            \textbf{(A)} Let $f_S = f_S(x) = \sum_{i \in S} x^i$ be the characteristic polynomial of $S$. Construct, in a similar fashion, the polynomial $f_T$ (for the set $T$) and let $g = f_S * f_T$. Observe that for $i\leq u$, the coefficient of $x^i$ in $g$ is nonzero if and only if $i \in S \oplus_u T$. Using \FFT, one can compute the polynomial $g$ in $O(u \log u)$ time, and extract $S \oplus_u T$ from it. 

            \textbf{(B)} Let $Z_1 = S_1$, and let $Z_i =Z_{i-1}  \oplus_u S_i$, for $i \in [2, k]$. Compute each $Z_i$, from $Z_{i-1}$ and $S_i$, in $O\bigl(u\log u\bigr)$ time using part (A). The total running time is $O\bigl( k\,u \log u \bigr)$. 

            \textbf{(C)} As in (A), let $f_S = f_S(x,y) = \sum_{(i,j) \in S} x^iy^j$ and $f_T$ be the characteristic polynomials of $S$ and $T$, respectively, and let $g = f_S * f_T$. For $i \leq u$ the coefficient of $x^iy^j$ is nonzero if and only if $(i,j)\in S \oplus_u T$. One can compute the polynomial $g$ by a straightforward reduction to regular \FFT (see multidimensional \FFT \cite[Chapter 12.8]{Blahut:1985:FAD:537283}), in $O\bigl( u\,v \log (u\,v)\bigr)$ time, and extract $S \oplus T$ from it.
        \end{proof} 
        
        The next lemma shows how to answer \textsc{AllSubsetSums}$^{\#}$ quickly, originally shown by the authors in \cite{koiliaris2017faster}, the proof is included for completeness. 

        \begin{lemma}[\textsc{AllSubsetSums}$^{\#}$ \cite{koiliaris2017faster}]
            \label{them:SSSC}
            Let $S \subseteq [u]$ be a given set of $n$ elements. One can compute, in $O\bigl( u \, n \log n \log u \bigr)$ time, the set $\SSSCu{S}$, which includes all subset sums of $S$ up to $u$ with cardinality information.
        \end{lemma}
        \begin{proof}
            Partition $S$ into two sets $S_1$ and $S_2$ of roughly the same size. Compute $\SSSCu{S_1}$ and $\SSSCu{S_2}$ recursively, and observe that $\SSSCu{S_1}$, $\SSSCu{S_2} \subseteq \bigl([u] \times \bigl[\frac{n}{2}\bigr]\bigr)$. Finally, note that $\SSSCu{S_1} \oplus_u \SSSCu{S_2} = \SSSCu{S}$. Applying Lemma~\hyperref[lem:l1]{\ref*{lem:l1}.\ref*{(C)}} yields $\SSSCu{S}$.
            
            The running time follows the recursive formula $T(n)=2\cdot T(n/2)+ O\bigl(u\, n \log u\bigr)$, which is $O\bigl(u\, n \log u\log n\bigr)$, proving the claim.  
        \end{proof}

        Next, we show how to compute the subset sums of elements in a congruence class quickly.

        \begin{lemma}
            \label{cor:FastSi}
            Let $\ell$, $b \in \mathbb{N}$ with $\ell < b$. Given a set $S \subseteq \set{ x \in \mathbb{N} \mid x \equiv \ell \pmod b }$ of size $n$, one can compute $\SSSu{S}$ in $O\bigl((u/b)\, n \log n \log u\bigr)$ time.
        \end{lemma}
        \begin{proof}
            An element $x \in S$ can be written as $x = y b + \ell$. Let $Q=\set{ y \mid yb+\ell \in S}$. As such, for any subset $X = \set{  y_1 b + \ell, \dots, y_j b + \ell } \subseteq S$ of size $j$, we have that
            \begin{equation*} 
            \sum_{x \in X} x%
            =%
            \sum_{i=1}^j (\,y_i b + \ell\,)
            =%
            \Bigl(\sum_{i=1}^j y_i\Bigr) b + j \ell
            \end{equation*}
            In particular, a pair $(z,j) \in \SSSCub{Q}$ corresponds
            to a set $Y=\set{y_1, \dots, y_j}\subseteq Q$ of size $j$, such that $\sum_i y_i = z$. The set $Y$ in turn corresponds to the set $X = \set{ y_1 b + \ell, \ldots, y_j b + \ell} \subseteq S$. By the above, the sum of the elements of $X$ is $ z b + j \ell$. 
            As such, compute $\SSSCub{Q}$, using the algorithm of Lemma~\ref{them:SSSC}, and return $\bigl\{\, z b + j \ell \bigm| (z,j) \in \SSSCub{Q} \,\bigr\} = \SSSu{S}$ as the desired result.
        \end{proof}

    \subsection{Algorithm}
        The new algorithm partitions the input into sets by congruence. Next it computes the \textsc{AllSubsetSums}$^{\#}$ for each such set, and combines the results. The algorithm is depicted in Figures~\ref{fig:alg2} and \ref{fig:alg1}.

        \begin{figure}[htbp]
            \begin{tikzpicture}
            \draw node[fill=black!12,inner sep=1ex,text width=\textwidth-2ex] {
                \underline{\textsc{AllSubsetSums}$^{\#}(S,u)$:}
                \vspace*{-.15cm} 
                \begin{tabbing}
                    {\small\textsf{INPUT:}} \quad \=A set $S$ of $n$ positive integers and an upper bound integer $u$.\\
                    {\small\textsf{OUTPUT:}} \>The set of all subset sums with cardinality information of $S$ up to $u$.
                \end{tabbing}
                \vspace*{-.45cm} 
                \begin{enumerate}
                    \item \textbf{if} $S=\set{x}$
                    \item \qquad \textbf{return} $\set{(0,0),(x,1)}$
                    \item $T \gets$ an arbitrary subset of $S$ of size $\lfloor n/2 \rfloor$ 
                    \item \textbf{return} $\textsc{AllSubsetSums}^{\#}(T,u) \oplus_u \textsc{AllSubsetSums}^{\#}(S \setminus T,u)$
                \end{enumerate}}; 
            \end{tikzpicture}
            \caption{The algorithm for the \textsc{AllSubsetSums}$^{\#}$ problem, used as a subroutine in Figure~\ref{fig:alg1}.}
            \label{fig:alg2}  
        \end{figure}
        \begin{figure}[htbp] 
            \begin{tikzpicture}
            \draw node[fill=black!12,inner sep=1ex,text width=\textwidth-2ex] {
                \underline{\textsc{AllSubsetSums}$(S,u)$:}
                \vspace*{-.15cm} 
                \begin{tabbing}
                    {\small\textsf{INPUT:}} \quad \=A set $S$ of $n$ positive integers and an upper bound integer $u$.\\
                    {\small\textsf{OUTPUT:}} \>The set of all realizable subset sums of $S$ up to $u$.
                \end{tabbing}
                \vspace*{-.15cm} 
                \begin{enumerate}
                    \item $b \gets \bigl\lfloor\sqrt{n \log n}\bigr\rfloor$
                    \item \textbf{for} $\ell\in [b-1]$ \textbf{do}
                    \item \qquad  $S_{\ell} \gets S \cap \set{ x \in \mathbb{N} \bigm| x \equiv \ell \pmod b }$
                    \item \qquad  $Q_{\ell} \gets \set{ (x-\ell)/b \bigm| x \in S_{\ell} }$
                    \item \qquad  $\SSSCub{Q_{\ell}} \gets$ \textsc{AllSubsetSums}$^{\#}\bigl(Q_{\ell}, \bigl\lfloor u/b\bigr\rfloor\bigr)$
                    \item \qquad  $R_{\ell} \gets \bigl\{\, z b + \ell j \bigm| (z,j)\in \SSSCub{Q_{\ell}}\,\bigr\}$ 
                    \item \textbf{return} $ R_0 \oplus_u \dots \oplus_u R_{b-1}$ 
                \end{enumerate}};
            \end{tikzpicture}
            \caption{The algorithm for \textsc{AllSubsetSums}.}
            \label{fig:alg1}
        \end{figure}

    \subsection{Result}

        \begin{theorem}[\textsc{AllSubsetSums}]
            \label{thm:mt}%
            Let $S \subseteq [u]$ be a given set of $n$ elements. One can compute, in $O\bigl( \sqrt{n \log n}\, u \log u \bigr)$ time, the set $\SSSu{S}$, which contains all subset sums of $S$ up to $u$.
        \end{theorem} 
        \begin{proof}
            Partition $S$ into $b = \lfloor\sqrt{n \log n}\rfloor$ sets $S_{\ell} = S$ $\cap$ $\set{x \in \mathbb{N} \mid x\equiv \ell \pmod b}$, $\ell \in [b-1]$, each of $n_{\ell}$ elements. For each $S_{\ell}$, compute the set of all subset sums $\SSSu{S_{\ell}}$ in $O\bigl((u/b)\, n_{\ell} \log n_{\ell} \log u\bigr)$ time by Lemma~\ref{cor:FastSi}. The time spent to compute all $\SSSu{S_{\ell}}$ is $\sum_{\ell\in [b-1]} O\bigl((u/b)\, n_{\ell} \log n_{\ell} \log u\bigr) = O\bigl((u/b)\, n \log n \log u\bigr)$. Combining $\SSSu{S_{0}} \oplus_u \dots \oplus_u \SSSu{S_{b-1}}$ using Lemma~\hyperref[lem:l1]{\ref*{lem:l1}.\ref*{(B)}} takes $O( b\, u \log u )$ time. Hence, the total running time is $O\bigl(\bigl(u/\lfloor\sqrt{n \log n}\rfloor\bigr)\, n \log n \log u + \lfloor\sqrt{n \log n}\rfloor\, u \log u \bigr) = O\bigl( \sqrt{n \log n}\, u \log u \bigr)$.
            \label{alg:end}    
        \end{proof}

        \begin{remark}
            \textsc{AllSubsetSums} is a generalization of \textsc{SubsetSum}, so the algorithm of \autoref{thm:mt} applies to it.
        \end{remark}

%
\bibliography{main}
\end{document}